\documentclass[envcountsame]{llncs}
\usepackage{ifpdf}
\ifpdf
\usepackage{pdfsync}
\fi
\usepackage{amsmath}
\usepackage{amssymb}
\usepackage{amsfonts}
\usepackage{mathtools}
\usepackage{mathdots}
\usepackage{mathrsfs}
\usepackage{enumerate}
\usepackage{gastex}

\usepackage{xcolor}
\usepackage{latexsym}

\usepackage{comment}
\usepackage{color}
\usepackage{graphicx}
\usepackage{epstopdf}

\usepackage{algpseudocode}
\usepackage{algorithm}

\usepackage{tikz}
\usetikzlibrary{arrows,automata,chains,matrix,scopes}

\makeatletter
\tikzset{join/.code=\tikzset{after node path={%
\ifx\tikzchainprevious\pgfutil@empty\else(\tikzchainprevious)%
edge[every join]#1(\tikzchaincurrent)\fi}}}
\makeatother
\tikzset{>=stealth',every on chain/.append style={join},
         every join/.style={->}}
\tikzstyle{labeled}=[execute at begin node=$\scriptstyle,
   execute at end node=$]

\numberwithin{equation}{section}

\algnewcommand{\LineComment}[1]{\State \(\triangleright\) #1}

\newtheorem{fact}{Fact}

\newcommand{\Between}{\mathsf{Btw}}
\newcommand{\broker}{\mathsf{BROKER}}
\newcommand{\Center}{\mathsf{Center}}
\newcommand{\CenterPeri}{\mathsf{CP}}
\newcommand{\Collab}{\mathsf{Col}}

\newcommand{\diam}{\mathsf{diam}}
\newcommand{\dist}{\mathsf{dist}}
\newcommand{\dominate}{\mathsf{DOM}}
\newcommand{\ecc}{\mathsf{ecc}}
\newcommand{\Enron}{\mathsf{Enron}}
\newcommand{\Facebook}{\mathsf{Facebook}}
\newcommand{\Max}{\mathsf{Max}}

\newcommand{\ML}{\mathsf{ML}}
\newcommand{\rad}{\mathsf{rad}}
\newcommand{\Improved}{\mathsf{Imp}}
\renewcommand{\restriction}{\mathord{\upharpoonright}}
\newcommand{\Simp}{\mathsf{S}}
\newcommand{\Periphery}{\mathsf{Periphery}}
\newcommand{\problemradius}{\mathsf{BROKER}}
\newcommand{\problemdiameter}{\mathsf{DIAM}}

\title{How to Build Your Network? A Structural Analysis}
\author{Anastasia Moskvina\inst{1} and  Jiamou Liu\inst{2}}
\institute{\ \inst{1} Auckland University of Technology, New Zealand\\
\texttt{anastasia.moskvina@aut.ac.nz}\\
\ \inst{2} The University of Auckland, New Zealand\\
\texttt{jiamou.liu@auckland.ac.nz}}

\begin{document}

\maketitle

\begin{abstract}
Creating new ties in a social network facilitates knowledge exchange and affects positional advantage.
In this paper, we study the process, which we call {\em network building}, of establishing ties between two existing social networks in order to reach certain structural goals. We focus on the case when one of the two networks consists only of a single member and motivate this case from two perspectives. The first perspective is {\em socialization}: we ask how a newcomer can forge relationships with an existing network to place herself at the center. We prove that obtaining optimal solutions to this problem is NP-complete, and present several efficient algorithms to solve this problem and compare them with each other. The second~perspective is {\em network expansion}: we investigate how a network may preserve or reduce its diameter through linking with a new node, hence ensuring small distance between its members. 
For both perspectives the experiment demonstrates that a small number of new links is usually sufficient to reach the respective goal.
\end{abstract}
\section{Introduction}
The creation of interpersonal ties has been a fundamental question in the structural analysis of social networks. While strong ties emerge between individuals with similar social circles, forming a basis of trust and hence community structure, weak ties link two members who share few common contacts. The influential work of Granovetter reveals the vital roles of weak ties: It is weak ties that enable information transfer between communities and provide individuals positional advantage and hence influence and power \cite{Granovetter}.

\smallskip

Natural questions arise regarding the establishment of weak ties between communities: How to merge two departments in an organization into one? How does a company establish trade with an existing market? How to create a transport map from existing routes? We refer to such questions as {\em network building}. The basic setup  involves two networks; the goal is to establish ties between them to achieve certain desirable properties in the combined network. A real-life example of network building is the inter-marriages between members of  the Medici, the leading family of Renaissance Florence, and numerous other noble Florentine families, towards gaining power and control over the city \cite{Jackson-Medici}. Another example is by Paul Revere, a prominent Patriot during the American Revolution, who strategically created social ties to raise a militia \cite{UzziDunlap}. 

\smallskip


The examples of the Medici and Paul Revere pose a more restricted scenario of network building: Here one of the two networks involved is only a single node, and the goal is to establish this node in the other network. We motivate this setup from two directions:
\begin{enumerate}
\item This setup amounts to the problem of {\em socialization}: the situation when a newcomer joins a network as an organizational member. A natural question for the newcomer is the following: How should I forge new relationships in order to take an advantageous position in the organization? As indicated in \cite{Morrison},  socialization is greatly influenced by the social relations formed by the newcomer with ``insiders'' of the network.
\item This setup also amounts to the problem of {\em network expansion}. For example, an airline expands its existing route map with a new destination, while trying to ensure a small number of legs between any cities. 
\end{enumerate}

{\em Distance} refers to the length of a shortest path between two members in a network; this
is an important measure of the amount of influence one may exert to another in the network \cite{collaboration}. The {\em radius} of a network refers to the maximal distance from a central member to all others in a network. Hence when a newcomer joins an established network,  it is in the interest of the newcomer to keep her distance to others bounded by the radius. The {\em diameter} of a network refers to the longest distance between any two members. It has long been argued from network science that small-world property -- the property that any two members of a network are linked by short paths -- improves network  robustness and facilitates information flow \cite{robustness}.
Hence it is in the interest of the  network to keep the diameter small as the network expands.
Furthermore, each relation requires time and effort to establish and maintain; thus one is interested in minimizing the number of new ties while building a network.

\paragraph*{\bf Contribution.} The novelty of this work is in proposing a formal, algorithmic study of organizational socialization. More specifically we investigate the following {\em network building problems}: Given a network $G$, add a new node $u$ to $G$ and create as few ties as possible for $u$ such that:
\begin{enumerate}
\item[(1)] $u$ is in the center of the resulting network; or
\item[(2)] the diameter of the resulting network is not larger than a specific value.
\end{enumerate}
Intuitively,  (1) asks how a newcomer $u$ may optimally connect herself with members of $G$, so that she belongs to the center. We prove that this problem is in fact NP-complete (Theorem~\ref{thm:problemradius}). Nevertheless, we give several efficient algorithms for this problem; in particular, we demonstrate a ``simplification'' process that significantly improves performance.
Intuitively, (2) asks how a network may preserve or reduce its diameter by connecting with a new member $u$. We show that ``preserving the diameter'' is trivial for most real-life networks and give two algorithms for ``reducing the diameter''.
We experimentally test and compare the performance of all our algorithms. Quite surprisingly, the experiments demonstrate that a very small number of new edges is usually sufficient for each problem even when the graph becomes large. 

\paragraph*{\bf Related works.} This work is predated by organizational behavioral studies \cite{socialization1,socialization2,Morrison}, which look at how social ties affect a newcomer's integration and assimilation to the organization. The authors in \cite{CrossThomas,UzziDunlap} argue {\em brokers} -- those who bridge and connect to diverse groups of individuals -- enable good network building; creating ties with and even becoming a broker oneself allows a person to gain private information, wide skill set and hence power. Network building theory has also been applied to various other contexts such as economics (strategic alliance of companies) \cite{Stuart1}, governance (forming inter-government contracts) \cite{federalism},  and politics (individuals' joining of political movements) \cite{Passy}. Compared to these works, the novelty here is in proposing a formal framework of network building, which employs techniques from complexity theory and algorithmics.

This work is also related to two forms of network formation: {\em dynamic models} and {\em agent-based models}, both aim to capture the natural emergence of social structures \cite{Jackson-Medici}. The former originates from random graphs, viewing the emergence of ties as a stochastic process which may or may not lead to an optimal structure \cite{entangle}. The latter comes from economics, treating a network as a multiagent system where utility-maximizing nodes establish ties in a competitive setting \cite{strategicNF,JacksonSurvey}. Our work differs from network formation as the focus here is on calculated strategies that achieve desirable goals in the combined network.

\section{Networks Building: The Problem Setup}

We view a {\em network} as an undirected unweighted connected graph  $G = (V, E)$ where $V$ is a set of nodes and $E$ is a set of (undirected) edges on $V$. We denote an edge $\{u,v\}$ as $uv$. If $uv\in E$ then $v$ is said to be {\em adjacent} to $u$.
A {\em path} (of {\em length} $k$) is a sequence of nodes $u_0,u_1,\ldots,u_k$ where $u_iu_{i+1}\!\in\!E$ for any $0\!\leq\!i\!<\!k$. The {\em distance} between $u$ and $v$, denoted by $\dist(u,v)$, is the length of a shortest path from $u$ to $v$. The {\em eccentricity} of $u$ is the maximum distance from $u$ to any other node, i.e., $\ecc(u)=\max_{v\in V} \dist(u,v)$. The {\em diameter} of the network $G$ is $ \diam(G)=\max_{u\in V} \ecc(u)$. The {\em radius} $\rad(G)$ of $G$ is $\min_{u\in V} \ecc(u)$. The {\em center} of $G$ consists of those nodes that are closest to all other nodes; it is the set $C(G)\coloneqq \{u\in V\mid \ecc(u)=\rad(G)\}$.
\begin{definition}
Let $G=(V,E)$ be a network and $u$ be a node not in $V$.  For $S\subseteq V$, denote by $E_S$ the set of edges $\{uv\mid v\in S\}$. Define $G\oplus_S u$ as the graph $(V\cup \{u\}, E\cup E_S)$.
\end{definition}
We require that $S\!\neq\!\varnothing$ and thus $G\oplus_S u$ is a network built by incorporating $u$ into $G$. By \cite{UzziDunlap}, for a newcomer $u$ to establish herself in $G$ it is essential to identify {\em information brokers} who connect to diverse parts of the network.
Following this intuition, we make the following definition
 \begin{definition}
 A set $S\subseteq V$ is a {\em broker set} of $G$ if $\ecc(u)=\rad(G\oplus_S u)$; namely, linking with $S$ enables $u$ to get in the center of the network.
 \end{definition}
Formally, given a network $G=(V,E)$, the problem of {\em network building for $u$} means selecting a set $S\!\subseteq\!V$ so that the combined network $G\!\oplus_S\!u$ satisfies certain conditions. Moreover, the desired set $S$ should contain as few nodes as possible. We focus on the following two key problems: 
\begin{enumerate}
\item $\problemradius$: The set $S$ is a broker set.
\item $\problemdiameter_\Delta$:  The diameter $\diam\!(\!G\!\oplus_S\! u\!)\!\leq\!\Delta$ for a given $\Delta\leq \diam(G)$.
\end{enumerate}
Note that for any network $G$, if $u$ is adjacent to all nodes  in $G$, it will have eccentricity 1, i.e., in the network $G\oplus_V u$, $\ecc(u)\!=\!1\!=\!\rad(G\oplus_V u)$ and $\diam(G\oplus_V u)\!=\!2$. Hence a desired $S$ must exist for $\problemradius$ and $\problemdiameter_\Delta$ where $\Delta\geq 2$. In subsequent section we systematically investigate these two problems.

\section{How to Be in the Center? Complexity and Algorithms for $\problemradius$}

\subsection{Complexity}
We  investigate the computational complexity of the decision problem $\broker(G,k)$, which is defined as follows:
\begin{description}
\item[INPUT] A network $G=(V,E)$, and an integer $k\geq 1$
\item[OUTPUT] Does $G$ have a broker set of size $k$?
\end{description}

The $\broker(G,k)$ problem is trivial if  $G$ has radius 1, as then $V$ is the only broker set. When $\rad(G)>1$, we recall the following notion: A set of nodes $S\subseteq V$ is a  {\em dominating set} if every node not in $S$ is adjacent to at least one member of $S$. The \emph{domination number} $\gamma(G)$ is the size of a smallest dominating set for $G$. The $\dominate(G,k)$ problem concerns testing whether  $\gamma(G)\!\leq\!k$ for a given graph $G$ and input $k$; it is a classical NP-complete decision problem \cite{GareyJohnson}.


\begin{theorem}\label{thm:problemradius}
The $\broker(G,k)$ problem is NP-complete.
\end{theorem}
\begin{proof} The $\broker(G,k)$ problem is clearly in NP. Therefore we only show NP-hardness.
We present a reduction from $\dominate(G,k)$ to $\broker(G,k)$. Note that when $\rad(G)\!=\!1$, $\gamma(G)\!=\!1$.  Hence $\dominate(G,k)$ remains NP-complete if we assume $\rad(G)>1$. Given a graph $G=(V,E)$ where $\rad(G)>1$, we construct a graph $H$. The set of nodes in $H$ is $\{v_i\mid v\in V, 1\leq i\leq 3\}$. The edges of $H$ are  as follows:
\begin{itemize}
\item Add an edge $v_i v_{i+1}$ for every $v\in V$, $1\leq i<3$
\item Add an edge $v_1w_1$ for every $v,w\in V$
\item Add an edge $v_2 w_2$ for every edge $vw\in E$
\end{itemize}
Namely, for each node $v\in V$ we create three nodes $v_1,v_2,v_3$ which form a path. We link the nodes in $\{v_1\mid v\in V\}$ to form a complete graph, and nodes in $\{v_2\mid v\in V\}$ to form a copy of $G$.  Since $\rad(G)\geq 2$, for each node $v\in V$ there is $w\in V$ with $\dist(v,w)\geq 2$. Hence in $H$,  $\dist(v_3,w_3)\geq 4$, and $\dist(v_2,w_3)\geq 3$. As the longest distance from any $v_1$ to any other node is $3$, we have $\rad(H)=3$.

Suppose $S$ is a dominating set of $G$. If we add all edges $uv$ where $v\in D=\{v_2\mid v\in S\}$, $\ecc(u)=3=\rad(H\oplus_D u)$. Hence $D$ is a broker set for $H$. Thus the size of a minimal broker set of $H$ is at most the size of a minimal dominating set of $G$.
Conversely, for any set $D$ of nodes in $H$, define the {\em projection} $p(D) = \{v\mid v_i\in D \text{ for some } 1\leq i\leq 3\}$.
 Suppose $p(D)$ is not a dominating set of $G$. Then there is some $v\in V$ such that for all $w\in p(D)$, $\dist(v_2,w_2)\geq 2$. Thus if we add all edges in $\{ux\mid x\in D\}$, $\dist(u,v_3)\geq 4$. But then $\ecc(w_1)=3$ for any $w\in p(D)$. So $D$ is not a broker set. This shows that the size of a minimal dominating set of $G$ is at most the size of a minimal broker set.

The above argument implies that the size of a minimal broker set for $H$ coincides with the size of a minimal dominating set for $G$. This finishes the reduction and hence the proof. \qed
\end{proof}

\subsection{Efficient Algorithms}\label{subsec:radius_algorithms}
Theorem~\ref{thm:problemradius} implies that computing optimal solution of $\problemradius$ is computationally hard. Nevertheless, we next present a number of efficient algorithms that take as input a network $G=(V,E)$ with radius $r$ and output a small broker set $S$ for $G$. A set $S\subseteq V$ is called {\em sub-radius dominating} if for all $v\in V$ not in $S$, there exists some $w\in S$ with $\dist(v,w)<r$. Our algorithms are based on the following fact, which is clear from definition: 

\begin{fact}\label{fact:sub-radius dominating}Any sub-radius dominating set is also a broker set.
\end{fact}

\subsubsection{(a) Three greedy algorithms} 

We first present three greedy algorithms; each algorithm applies a heuristic that iteratively adds new nodes to the broker set $S$.
The starting configuration is $S= \varnothing$ and $U= V$.  During its computation, the algorithm maintains a subgraph $F=(U,E\restriction U)$, which is induced by the set $U$ of all ``uncovered'' nodes, i.e., nodes that have distance $>(r-1)$ from any current nodes in $S$. It repeatedly performs the following operations until $U=\varnothing$, at which point it outputs $S$:
\begin{enumerate}
\item Select a node $v\in U$ based on the corresponding heuristic and add $v$ to $S$. 
\item Compute all nodes at distance at most  $(r-1)$ from $v$. Remove these nodes and all attached edges from $F$. 
\end{enumerate}
\paragraph*{\bf Algorithm 1: $\Max$ (Max-Degree).} The first heuristic is based on the intuition that one should connect to the person with the highest number of social ties; at each iteration, it adds to $S$ a node with maximum degree in the graph $F$.

\paragraph*{\bf Algorithm 2: $\Between$ (Betweenness).} The second heuristic is based on {\em betweenness}, an important centrality measure in networks \cite{betweeness}. More precisely, the {\em betweenness} of a node $v$ is the number of shortest paths from all nodes to all others that pass through $v$. Hence high betweenness of $v$ implies, in some sense, that $v$ is more likely to have short distance with others. This heuristic works in the same manner as $\Max$ but picks nodes with maximum betweenness in $F$.

\paragraph*{\bf Algorithm 3: $\ML$ (Min-Leaf).}  The third heuristic is based on the following intuition: A node is called a {\em leaf} if it has minimum degree in the network; leaves correspond to  least connected members in the network, and may become outliers once nodes with higher degrees are removed from the network. Hence this heuristic gives first priority to leaves. Namely, at each iteration, the heuristic adds to $S$ a node that has distance at most $r-1$ from $v$. More precisely, the heuristic first picks a leaf $v$ in $F$, then applies a sub-procedure to find the next node $w$ to be added to $S$. The sub-procedure  determines a path $v=u_1,u_2,\ldots$ in $F$ iteratively as follows:
    \begin{enumerate}
    \item Suppose $u_i$ is picked. If $i=r$ or $u_i$ has no adjacent node in $F$, set $u_i$ as $w$ and terminate the process.
    \item Otherwise select a $u_{i+1}$ (which is different from $u_{i-1}$) among adjacent nodes of $u_i$ with maximum degree.
    \end{enumerate}
    After the process above terminates, the algorithm adds $w$ to $S$. Note that the distance between $w$ and $v$ is at most  $r-1$.

We mention that Algorithms 1,3 have been applied in \cite{k-domination} to {\em regular graphs}, i.e., graphs where all nodes have the same degree. In particular, $\ML$ has been shown to produce small $k$-dominating sets for given $k$ in the average case for regular graphs. 


\subsubsection{(b) Simplified greedy algorithms}

One significant shortcoming of Algorithms 1--3 is that, by deleting nodes from the network $G$, the network may become disconnected, and nodes that could have been connected via short paths are no longer reachable from each other. This process may produce {\em isolated} nodes in $F$, i.e., nodes having degree 0, which are subsequently all added to the output set $S$. Moreover, maintaining the graph $F$ at each iteration also makes  implementations more complex. Therefore we next propose {\em simplified} versions of  Algorithms  1--3. 

\paragraph*{\bf Algorithms 4 $\Simp$-$\Max$, 5 $\Simp$-$\Between$, 6 $\Simp$-$\ML$.} The simplified algorithms act in a similar way as their ``non-simplified'' counterparts; the difference is that here the heuristic works over the original network $G$ as opposed to the updated network $F$. Hence the graph $F$ is no longer computed. Instead we only need to maintain a set $U$ of ``uncovered'' nodes. The simplified algorithms have the following general structure: \ Start from $S= \varnothing$ and $U= V$, and repeatedly perform the following until $U=\varnothing$, at which point output $S$:
\begin{enumerate}
\item Select a node $v$ from $U$ based on the corresponding heuristic and add $v$ to $S$.
\item Compute all nodes with distance $<\rad(G)$ from $v$, and remove any of these node  from $U$.
\end{enumerate}
We stress that here the same heuristics as described above in Algorithms 1--3 are applied, except that we replace any mention of ``$F$'' in the description with ``$U$'', while all notions of degrees, distances, and betweenness are calculated based on the original network $G$.

As an example, in Fig.~\ref{fig:Max} we run $\Max$ and $\Simp$-$\Max$ on the same network $G$, which contains 30 nodes. The figures show the result of both algorithms, and in particular, how $\Simp$-$\Max$ outputs a smaller sub-radius dominating set. We further verify via experiments below that the simplified algorithms lead to much smaller output $S$ in almost all cases.

\begin{figure}[!] \centering
        \includegraphics[width=0.5\textwidth]{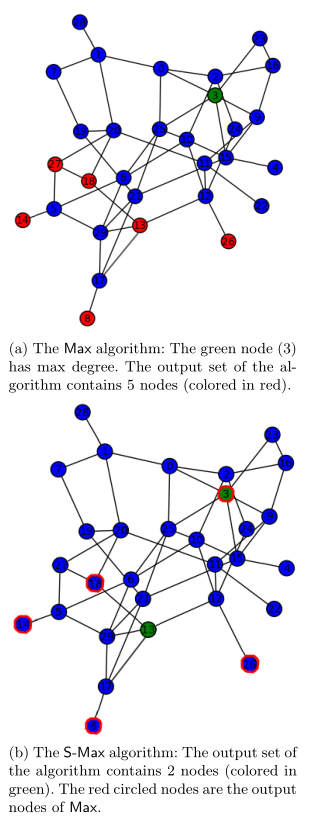}        
        \caption{ The network $G$ contains 30 nodes and has radius $\rad(G)=4$. The $\Max$ algorithm: The algorithm first puts node 3 (shown in green) into $S$. Then removes all nodes (and attached edges that are at distance three from the node 3; these nodes are considered ``covered'' by 3. In the remaining graph, there are three isolated nodes 8,14,26, as well as a line of length 2. The algorithm then puts the node 18 into $S$ which ``covers'' 27 and 13. Thus the output set is $S=\{3,18,8,14,26\}$. The $\Simp$-$\Max$ algorithm: The algorithm first puts 3 into the set $S$, but does not remove the covered nodes. It simply construct a set containing all ``uncovered'' nodes, namely, $\{27,18,13,14,8,26\}$. The algorithm then selects the node 13 which has max degree from these nodes, and puts into $S$. It then turns out that all nodes are covered. Therefore the output set is $S=\{3,13\}$. Thus $\Simp$-$\Max$ is superior in this example.}\label{fig:Max}
\end{figure}

\subsubsection{(c) Center-based algorithms}

The 6 algorithms presented above can all be applied to find $k$-dominating set for arbitrary $k\geq 1$. Since our focus is in finding sub-radius dominating set to answer the $\problemradius$ problem, we describe two algorithms that are specifically designed for this task. When building network for a newcomer, it is natural to consider nodes that are already in the center of the network $G$. Hence our two algorithms are based on utilizing the center of $G$.

\paragraph*{\bf Algorithm 7 $\Center$.} The algorithm finds a center $v$ in $G$ with minimum degree, then output all nodes that are adjacent to $v$. Since $v$ belongs to the center, for all $w\in V$, we have $\dist(v,w)\leq \rad(G)$ and thus there is $v'$ adjacent to $v$ such that  $\dist(w,v')=\dist(w,v)-1<\rad(G)$. Hence the algorithm returns a sub-radius dominating set. Despite its apparent simplicity, $\Center$ returns surprisingly good results in many cases, as shown in the experiments below.

\paragraph*{\bf Algorithm 8 $\Improved$-$\Center$.} We present a modified version of $\Center$, which we call $\Improved$-$\Center$.
The algorithm first picks a center with minimum degree, and then orders all its neighbors in decreasing degree. It adds the first neighbor to $S$ and remove all nodes $\leq (r-1)$-steps from it. This may disconnect the graph into a few connected components. Take the largest component $C$. If $C$ has a smaller radius  than $r$, we add the center of this component to $S$; otherwise we add the next neighbor to $S$. We then remove from $F$ all nodes at distance $\leq (r-1)$ from the newly added node. This procedure is repeated until $F$ is empty. See Procedure~\ref{alg:Center_improved}. Fig.~\ref{fig:Center} shows an example where $\Improved$-$\Center$ out-performs $\Center$.

\begin{algorithm}[!htb]
\floatname{algorithm}{Procedure}
 \caption{ $\Improved$-$\Center$: Given $G=(V,E)$ (with radius $r$)}
 \label{alg:Center_improved}
  \begin{algorithmic}
    \State Pick a center node $v$ in $G$ with minimum degree $d$
    \State Sort all adjacent nodes of $v$ to a list $u_1,u_2,\ldots,u_d$ in decreasing order of degrees
    \State Set $S\leftarrow \varnothing$ and $i\leftarrow 1$
    \While{$U\neq \varnothing$}
        \State Set $C$ as the largest connected component in $F$

        \If{$\rad(C)<\rad(G)-1$}
            \State Pick  a center node $w$ of $C$. Set $S\leftarrow S\cup \{w\}$
            \State Set $U\leftarrow U\setminus \{w'\in U\mid \dist(w,w')<r\}$
        \Else
            \State Set $S\leftarrow S\cup \{u_i\}$
            \State Set $U\leftarrow U\setminus \{w'\in U\mid \dist(u_i,w')<r\}$
            \State Set $i\leftarrow i+1$
        \EndIf
        \State Set $F$ as the subgraph induced by the current $U$
    \EndWhile
    \State\Return $S$
  \end{algorithmic}
\end{algorithm}
\begin{figure}[!htb]
        \centering
                \includegraphics[width=0.7\textwidth]{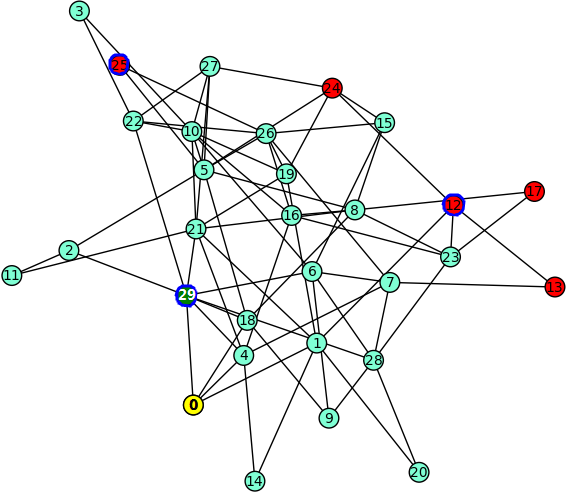}
        \caption{ The graph $G$ has radius $\rad(G)=3$. The yellow node 0 is a center with min degree 4. Thus $\Center$ outputs 4 nodes $\{1,4,18,29\}$. The dark green node 29 adjacent to 0 has max degree; the red nodes are ``uncovered'' by 29. Thus $\Improved$-$\Center$ outputs the 3 blue circled nodes $\{12,25,29\}$. }\label{fig:Center} 
\end{figure}
Finally, we note that all of Algorithms 1--8 output a sub-radius dominating set $S$ for the network $G$. Thus the following theorem is a direct implication from Fact~\ref{fact:sub-radius dominating}.
\begin{theorem}
All of Algorithms 1--8 output a brocker set for the network $G$.
\end{theorem}
\subsection{Experiments for $\problemradius$}
We implemented the algorithms using Sage \cite{Sage}. 
We apply two models of random graphs: The first (BA) is Barabasi-Albert's preferential attachment model which generates scale-free graphs whose degree distribution of nodes follows a power law; this is an essential property of numerous real-world networks \cite{BA}. The second (NWS) is Newman-Watts-Strogatz's small-world network \cite{NewmanWattsStrogatz}, which produces graphs with small average path lengths and high clustering coefficient. 

For each algorithm we are interested in two indicators of its performance: 1) {\em Output size}: The average size of the output broker set (for a specific class of random graphs). 2) {\em Optimality rate}: The probability that the algorithm gives optimal broker set for a random graph. To compute this we need to first compute the size of an optimal broker set (by brute force) and count the number of times the algorithm produces optimal solution for the generated graphs.

\paragraph*{\bf Experiment 1: Output sizes.} We generate $300$ graphs whose numbers of nodes vary between $100$ and $1000$ using each random graph model. We compute averaged output sizes of generated graphs by their number of nodes $n$ and radius $r$. The results are shown in Fig.~\ref{fig:improvedRes}. From the result we see: a) The simplified algorithms produce significantly smaller broker sets compared to  their unsimplified counterparts. This shows superiority of the simplified algorithms.
b) BA graphs in general allow smaller output set than NWS graphs. This may be due to the scale-free property which results in high skewness of the degree distribution.

\begin{figure}[!htb]
        \centering 
	   \includegraphics[width=\textwidth]{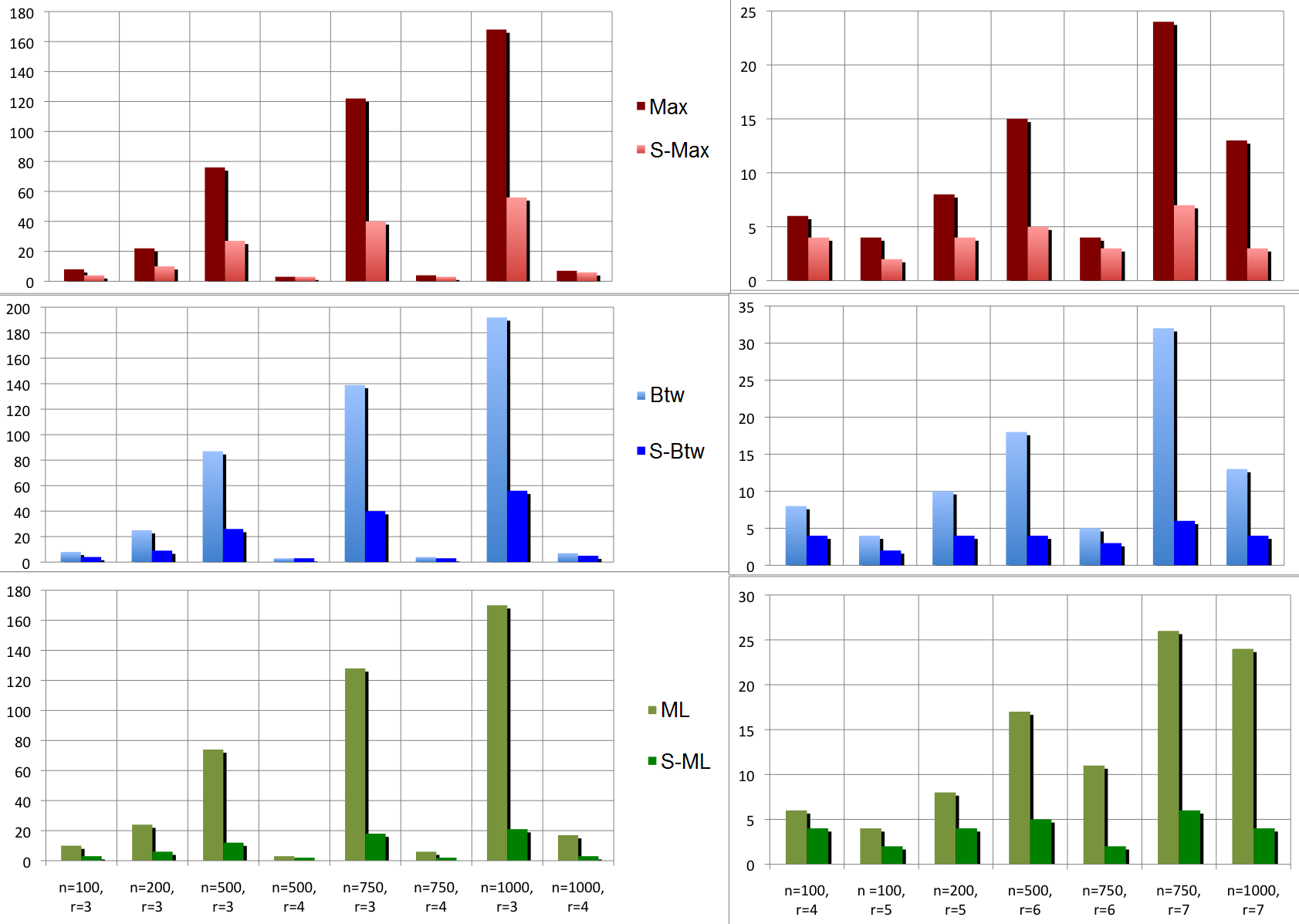}
        \caption{ Comparing results: average performance of the $\Max$, $\Between$, $\ML$, algorithms versus their simplified versions on randomly generated graphs (BA graphs on the left; NWS on the right)}\label{fig:improvedRes}
\end{figure}
\paragraph{\bf Experiment 2: Optimality rates.} For the second goal, we compute the optimality rates of algorithms when applied to random graphs, which are shown in Fig.~\ref{fig:res}. For BA graphs, the simplified algorithm $\Simp$-$\ML$ has significantly higher optimality rate ($\geq 85\%$) than other algorithms. On the contrary, its unsimplified counterpart $\ML$ has the worst optimality rate. This is somewhat contrary to Duckworth and Mans's work showing $\ML$ gives very small solution set for regular graphs \cite{k-domination}. 
For NWS graphs, several algorithms have almost equal optimality rate.  The three best algorithms are $\Simp$-$\Max$, $\Simp$-$\Between$ and $\Simp$-$\ML$ which has varying performance for graphs with different sizes (See Fig.~\ref{fig:optimality_nodes}).

    \begin{figure}[!htb]\centering

                \includegraphics[width=\textwidth]{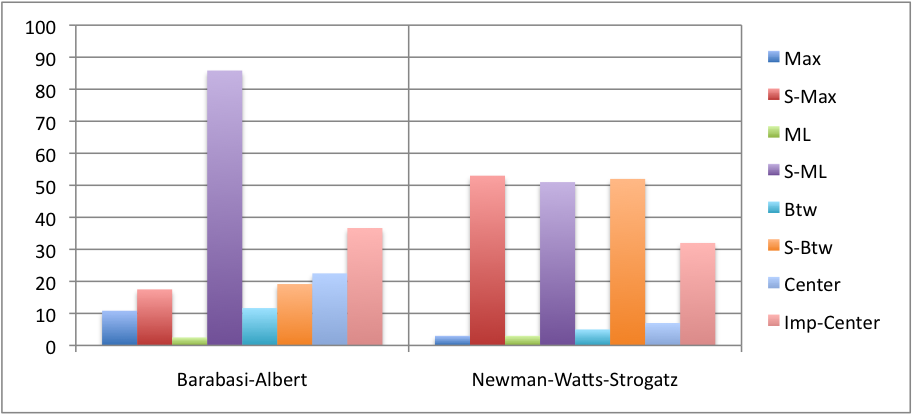}
        \caption{ Optimality rates for different types of random graphs}\label{fig:res} 
    \end{figure}

    \begin{figure}[!htb]\centering
                \includegraphics[width=\textwidth]{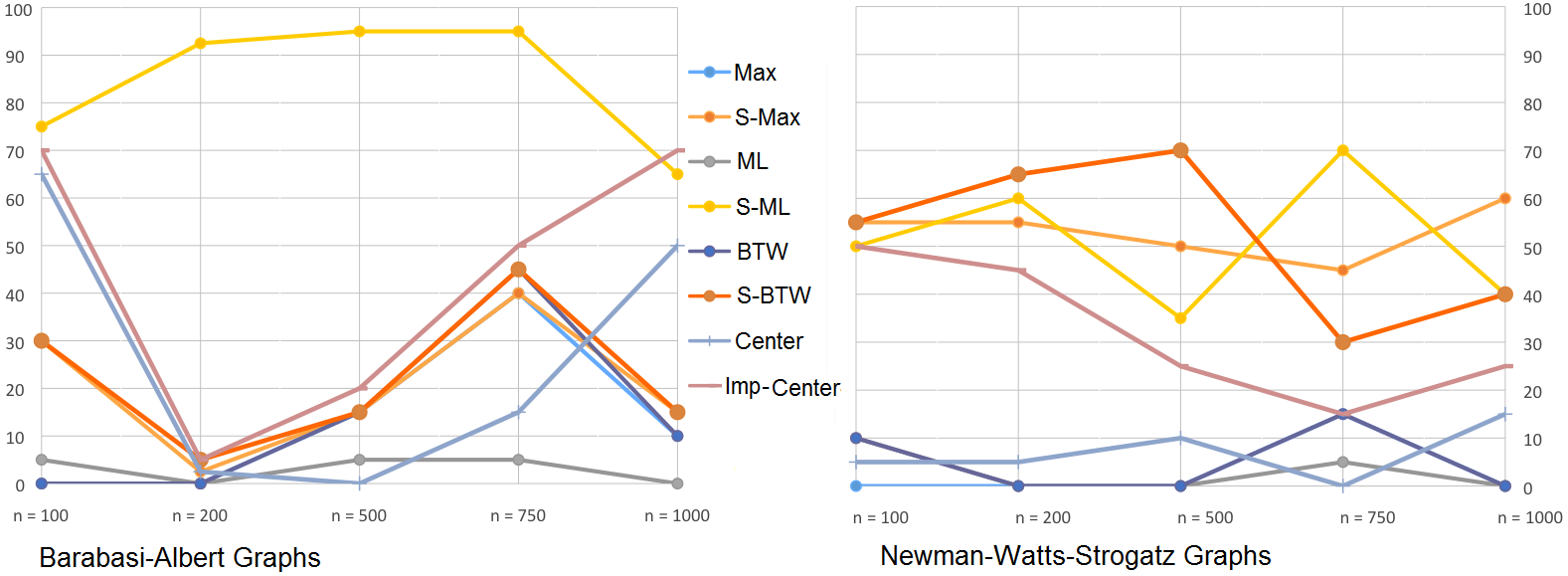}
        \caption{ Optimality rates when graphs are classified by sizes}\label{fig:optimality_nodes}
    \end{figure}
\paragraph*{\bf Experiment 3: Real-world datasets.}
We test the algorithms on several real-world datasets: The $\Facebook$ dataset, collected from survey participants of Facebook App, consists of friendship relation on Facebook \cite{facebook}. $\Enron$  is an email network  of the company made public by the FERC \cite{enron1}. Nodes of the network are email addresses and if an address $i$ sent at least one email to address $j$, the graph contains an undirected edge from $i$ to $j$.
$\Collab 1$ and $\Collab 2$ are collaboration networks that represent scientific collaborations between authors papers submitted to General Relativity and Quantum Cosmology category ($\Collab 1$), and to High Energy Physics Theory category ($\Collab 2$) \cite{collaboration}. 

\begin{table}[!htb]
\centering
\begin{tabular}{| l | l | l| l | l |}
\hline
&  Facebook &  Enron &  Col1 &  Col2\\
\hline
 Number of nodes &  4,039	 &  33,969 &  4,158&	 8,638	\\
\hline
 Number of edges &  88,234&  180,811 & 13,422& 24,806\\
\hline
Largest connected subgraph & 4,039 & 33,696 & 4,158 &	8,638 \\
\hline
 Diameter &  8 &  13 &  17 &  18\\
\hline
 Radius &  4 &  7 &  9 &  10\\
\hline
  \end{tabular}\caption{ Network properties} \label{table:datasets} 
\end{table}

Results on the datasets are shown in Fig.~\ref{fig:datasets_res}.  $\Between$ and $\Simp$-$\Between$ algorithms become too inefficient as it requires computing shortest paths between all pairs in each iteration. Moreover, $\Simp$-$\Max$ also did not terminate within reasonable time for the $\Enron$ dataset. Even though the datasets have many nodes, the output sizes are in fact very small (within 10). For instance, the smallest output sets of the $\Enron$, $\Collab 1$ and $\Collab 2$ contain just two nodes. In some sense, it means that to become in the center even in a large social network, it is often enough to establish only very few connections.

\begin{figure}[!htb]
        \centering
                \includegraphics[width=.8\textwidth]{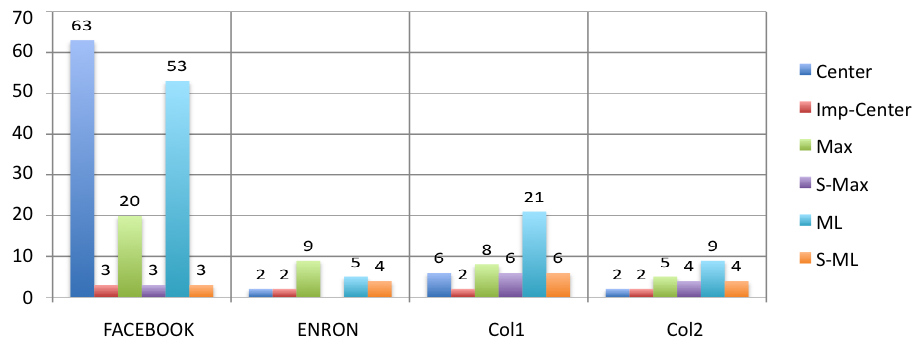}
        \caption{ The number of new ties  for the four real-world networks}\label{fig:datasets_res} 
\end{figure}

Among all algorithm $\Improved$-$\Center$ has the best performance, producing the smallest output set for all networks. Moreover, for $\Enron$, $\Collab 1$ and $\Collab 2$,  $\Improved$-$\Center$  returns the optimal broker set with cardinality $2$. A rather surprising fact is, despite straightforward seemingly-naive logic,  $\Center$ also produces small outputs in three networks. This reflects the fact that in order to become central it is often a good strategy to create ties with the friends of a central person. 
\section{How to Preserve or Improve the Diameter? Complexity and Algorithms for $\problemdiameter_\Delta$}\label{sec:diameter}
Let $G=(V,E)$ be a network and $u\notin V$. The $\problemdiameter_\Delta$ problem asks for a set $S\subseteq V$ such that the network $G\oplus_S u$  has diameter $\leq\Delta$; we refer to any such $S$ as {\em $\Delta$-enabling}.

\subsection{Preserving the diameter}
We first look at a special case when $\Delta=\diam(G)$, which has a natural motivation: How can an airline expand its existing route map with an additional destination while ensuring the maximum number of hops between any two destinations is not increased? We are interested in creating as few new connections as possible to reach this goal. Let $\delta(G)$ denote the size of the smallest $\diam(G)$-enabling set for $G$. We say a graph is {\em diametrically uniform} if all nodes have the same eccentricity. 


\begin{theorem}\label{thm:problemdiameter}
\begin{enumerate}
\item[(a)]If $G$ is not diametrically uniform,$\delta(G)\!=\!1$.
\item[(b)]If $G$ is complete, then $\delta(G)=|V|$.
\item[(c)]If $G$ is diametrically uniform and incomplete, then $1<\delta(G)\leq d$ where $d$ is the minimum degree of any node in $G$, and the upper bound $d$ is sharp.
\end{enumerate}
\end{theorem}
\begin{proof}
For {\bf (a)}, suppose $G$ is not diametrically uniform. Take any $v$ where $\ecc(v)<\diam(G)$. Then in the expanded network $G\oplus_{\{v\}} u$, we have $\ecc(u)=\ecc(v)+1\leq \diam(G)$. {\bf (b)} is clear.
For {\bf (c)} Suppose $G$ is diametrically uniform and incomplete. For the lower bound, suppose $\gamma_{\diam(G)-1}(G)=1$. Then there is some $v\in V$ with the following property: In the network $G\oplus_{\{v\}} u$ we have $\ecc(u)\leq \diam(G)$, which means that $\ecc(v)<\diam(G)$. This contradicts the fact that $G$ is  diametrically uniform. For the upper bound, take a node $v\in V$ with the minimum degree $d$. Let $N$ be the set of nodes adjacent to $v$. From any node $w\neq v$, there is a shortest path of length $\leq \diam(G)$ to $v$. This path contains a node in $N$. Hence $w$ is at distance $\leq \diam(G)-1$ from some node in $N$. Furthermore as $G$ is not complete, $\diam(G)\geq 2$ and $v$ is at distance $1\leq \diam(G)-1$ from nodes in $N$. \qed
\end{proof}
\paragraph*{Remark} We point out that in case (c) calculating the exact value of $\delta(G)$ is a hard: In \cite{hardness_diameter}, its parametrized complexity  is shown to be complete for $\mathsf{W}[2]$, second level of the $\mathsf{W}$-hierarchy. Hence $\problemdiameter_\Delta$ is unlikely  to be in $\mathsf{P}$. On the other hand, we argue that real-life networks are rarely diametrically uniform. Hence by Thm.~\ref{thm:problemdiameter}(a), the smallest number of new connections needed to preserve the diameter is 1.

\subsection{Reducing the diameter}
We now explore the question $\problemdiameter_\Delta$ where $2\leq \Delta<\diam(G)$; this refers to the goal of placing a new member in the network and creating ties to allow a closer distance between all pairs of members.   We suggest two heuristics to solve this problem. 

\paragraph*{\bf Algorithm 9 $\Periphery$.} The {\em periphery} $P(G)$ of $G$ consists of all nodes $v$ with $\ecc(v)=\diam(G)$. Suppose $\diam(G)>2$. Then the combined network $G\oplus_{P(G)} u$ has diameter smaller than $\diam(G)$. Hence we apply the following heuristic:  Two nodes $v,w$ in $G$ are said to form a {\em peripheral pair} if $\dist(v,w)=\diam(G)$. The algorithm first adds the new node $u$ to $G$ and repeats the following procedure until the current graph has diameter $\leq \Delta$:\\
1) Randomly pick a peripheral pair $v,w$ in the current graph\\
2) Adds the edges $uv,uw$ if they have not been added already\\
3) Compute the diameter of the updated graph

\noindent Note that once $v,w$ are chosen as a peripheral pair and the corresponding edges $uv,uw$ added, $v$ and $w$ will have distance 2 and they will not be chosen as a peripheral pair again. Hence the algorithm eventually terminates and produces a graph with diameter at most $\Delta$.

\paragraph*{\bf Algorithm 10 $\CenterPeri$ (Center-Periphery).} This algorithm applies a similar heuristic as $\Periphery$, but instead of picking peripheral pairs at each iteration, it first picks a node $v$ in the center and adds the edge $uv$; it then repeats the following procedure until the current graph has diameter $\leq \Delta$:\\
1) Randomly pick a node $w$ in the periphery of the current graph\\
2) Add the edge $uw$ if it has not been added already\\
3) Compute the diameter of the updated graph

\noindent Suppose at one iteration the algorithm picks $w$ in the periphery. Then after this iteration the eccentricity of $w$ is at most $r+2$ where $r$ is the radius of the graph.

\subsection{Experiments for $\problemdiameter_\Delta$}
We implement and test the performance of Algorithms 9,10 for  the problem $\problemdiameter_\Delta$.
The performance of these algorithms are measured by the number of new ties created.


\paragraph*{\bf Experiment 4: Random graphs.}
We apply the two models of random graphs, BA and NWS, as described above. We generated $350$ graphs and considered the case when $\Delta = d(G) - 1$, i.e. the aim was to improve the diameter by one.
For both types of random graphs (fixing size and radius), the average number of new ties are shown in Fig.~\ref{fig:Barabasi_improve}.
The experiments show that $\Periphery$ performs better when the radius of the graph is close to the diameter (when radius is $>2/3$ of diameter), whilst $\CenterPeri$ is slightly better when the radius is significantly smaller than the diameter.

\begin{figure}[!htb]
        \centering
                \includegraphics[width=\textwidth]{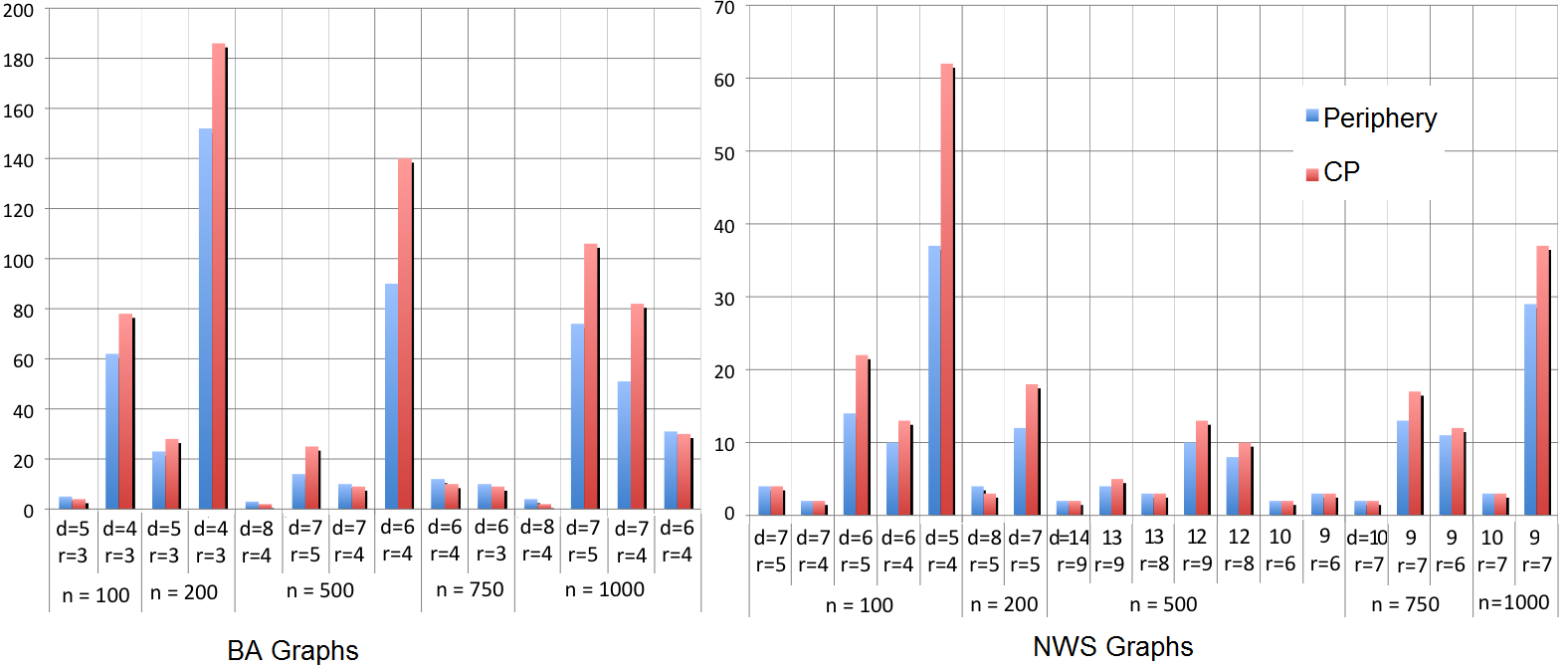}
        \caption{ Comparing two methods for improving diameter applied to BA (left) and NWS (right) graphs}\label{fig:Barabasi_improve}
\end{figure}

\paragraph*{\bf Experiment 5: Real-World Datasets.} We run both $\Periphery$ and $\CenterPeri$ on the networks $\Collab 1$ and $\Collab 2$  introduced above, setting $\Delta =\diam(G)-i$ for $1\leq i\leq 4$. The numbers of new edges obtained by $\Periphery$ and $\CenterPeri$ are shown in Figure~\ref{fig:collaboration_improve}; naturally for increasing $i$, more ties need to be created. 
We point out that, despite the large total number of nodes, one needs less than $19$ new edges to improve the diameter even by four.
This reveals an interesting phenomenon: While a collaboration network may be large, a few more collaborations are sufficient to reduce the diameter of the network.

On the $\Facebook$ dataset, $\Periphery$ is significantly better than $\CenterPeri$: To reduce the diameter of this network from $8$ to $7$, $\Periphery$ requires 2 edges while $\CenterPeri$ requires $47$. When one wants to reach the diameter $6$, the numbers of new edges increase to 6 for $\Periphery$ and 208 for $\CenterPeri$.

\begin{figure}[!htb]
        \centering
                \includegraphics[width=0.8\textwidth]{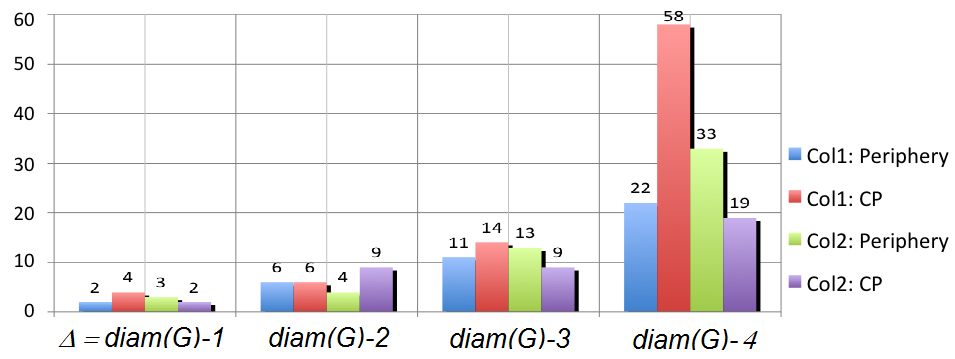}
        \caption{ Applying algorithms for improving diameter to Collaboration 1 and Collaboration 2 datasets}\label{fig:collaboration_improve}
\end{figure}

\section{Conclusion and Outlook}
This work studies how ties are built between a newcomer and an established network to reach certain structural properties. Despite achieving optimality is often computationally hard, there are efficient heuristics that reach the desired goals using few new edges. We also observe that the number of new links required to achieve the specified properties remain small even for large networks.

This work amounts to an effort towards an algorithmic study of network building. Along this effort, natural questions have yet to be explored include: (1) Investigating the creation of ties between two arbitrary networks, namely, how ties are created between two established networks to maintain or reduce diameter. (2) When building networks in an organizational context (such as merging two departments in a company), one normally needs not only to take into account the informal social relations, but also formal ties such as the reporting relations, which are typically directed edges \cite{LiuMoskvina}. We plan to investigate network building in an organizational management perspective by incorporating both types of ties.

\bibliographystyle{splncs03}

\bibliography{ijcai16}

\begin{thebibliography}{10}
\providecommand{\url}[1]{\texttt{#1}}
\providecommand{\urlprefix}{URL }

\bibitem{federalism}
Andrew, S.A.: Adaptive versus restrictive contracts: Can they resulve different
  risk problems? In: Feiock, R., Scholz, J. (eds.) Self-Organizing Federalism:
  Collaborative Mechanisms to Mitigate Institutional Collective Action
  Dilemmas. Cambridge University Press (2010)

\bibitem{BA}
Barb{\'a}si, A.L., Albert, R.: Emergence of scaling in random networks. Science
   286(5439),  509--512 (Oct 1999)

\bibitem{betweeness}
Barth{\`e}lemy, M.: Betweenness centrality in large complex networks. Eur.
  Phys. J. B  38,  163--168 (2004)

\bibitem{CrossThomas}
Cross, R., Thomas, R.: Managing yourself: a smarter way to network. Harvard
  Business Review  89(7--8),  149--153 (Jul--Aug 2011)

\bibitem{entangle}
Donetti, L., Hurtado, P.I., Munoz, M.A.: Entangled networks, synchronization
  and optimal network topology. Phys. Rev. Lett.  95(188701) (2005)

\bibitem{k-domination}
Duckworth, W., Mans, B.: Randomized greedy algorithms for finding small
  $k$-dominating sets of regular graphs. Random Structures and Algorithms
  27(3),  401--412 (2005)

\bibitem{GareyJohnson}
Garey, M.R., Johnson, D.S.: Computers and Intractability: A Guide to the Theory
  of NP-Completeness. W.H.Freeman (1979)

\bibitem{Granovetter}
Granovetter, M.S.: The strength of weak ties. The American Journal of Sociology
   78(6),  1360--1380 (1973)

\bibitem{socialization2}
Jablin, F.M., Krone, K.J.: Organizational assimilation. In: Berger, C.,
  Chaffee, S. (eds.) Handbook of communication science, pp. 711--–746. Sage
  (1987)

\bibitem{JacksonSurvey}
Jackson, M.O.: A survey of models of network formation: Stability and
  efficiency. In: Demange, G., Wooders, M. (eds.) Group Formation in Economics;
  Networks, Clubs and Coalitions. Cambridge University Press (2004)

\bibitem{Jackson-Medici}
Jackson, M.O.: The economics of social networks. In: Blundell, R., Newey, W.,
  Persson, T. (eds.) Proceedings of the 9th World Congress of the Econometric
  Society. Cambridge University Press (2006)

\bibitem{strategicNF}
Kleinberg, J., Suri, S., Tardos, E., Wexler, T.: Strategic network formation
  with structural holes. ACM SIGecom Exchanges  7(3) (November 2008)

\bibitem{collaboration}
Leskovec, J., Kleinberg, J., Faloutsos, C.: Graph evolution: Densification and
  shrinking diameters. ACM Transactions on Knowledge Discovery from Data (ACM
  TKDD)  1(1) (2007)

\bibitem{enron1}
Leskovec, J., Lang, K.J., Dasgupta, A., Mahoney, M.: Community structure in
  large networks: Natural cluster sizes and the absence of large well-defined
  clusters. Internet Mathematics  6(1),  29--123 (2009)

\bibitem{LiuMoskvina}
Liu, J., Moskvina, A.: Hierarchies, ties and power in organizational networks:
  Model and analysis. In: ASONAM '15 Proceedings of the 2015 IEEE/ACM
  International Conference on Advances in Social Networks Analysis and Mining.
  pp. 202--209 (2015)

\bibitem{hardness_diameter}
Lokshtanov, D., Misra, N., Philip, G., Ramanujan, M.S., Saurabh, S.: Hardness
  of r-dominating set on graphs of diameter $(r-1)$. In: Proceeds of the 8th
  International Symposium Parameterized and Exact Computation (IPEC 2013). pp.
  255--267. Sophia Antipolis, France (September 2013)

\bibitem{facebook}
McAuley, J., Leskovec, J.: Learning to discover social circles in ego networks.
  In: The Twenty-sixth Annual Conference on Neural Information Processing
  Systems (2012)

\bibitem{Morrison}
Morrison, E.W.: Newcomers' relationships: The role of social network ties
  during socialization. The Academy of Management Journal  45(6),  1149--1160
  (2002)

\bibitem{NewmanWattsStrogatz}
Newman, M.E., Watts, D.J., Strogatz, S.H.: Random graph models of social
  networks. Proc. Nat. Acad. Sci. USA  99,  2566--2572 (2002)

\bibitem{Passy}
Passy, F.: Social networks matter. but how? In: Diani, M., McAdam, D. (eds.)
  Social movement and networks: relational approaches to collective action, pp.
  21--–48. Oxford University Press (2003)

\bibitem{socialization1}
Sherman, J., Smith, H.L., Mansfield, E.R.: The impact of emergent network
  structure on organizational socialization. Journal of Applied Behavioral
  Science  22,  53--–63 (1986)

\bibitem{Sage}
Stein, W.A.: Sage -- a computer system for algebra and geometry
  experimentation. Tech. rep. (2012), http://wstein.org/sage.html

\bibitem{Stuart1}
Stuart, T.E.: Network positions and propensities to collaborate: An
  investigation of strategic alliance formation in a high-technology industry.
  Administrative science quarterly  43(3),  668--698 (Sep 1998)

\bibitem{UzziDunlap}
Uzzi, B., Dunlap, S.: How to build your network. Harvard Business Review
  83(12),  53--60 (Dec 2005)

\bibitem{robustness}
Wang, X., Chen, G.: Complex networks: Small-world, scale-free and beyond. IEEE
  circuits and systems magazine pp. 6--20 (First Quarter 2003)

\end{thebibliography}

\end{document}